\documentclass[orivec]{llncs} 
\usepackage{hyperref}
\usepackage{float}
\usepackage{amsmath}

\usepackage{amsthm}
\usepackage{algorithm, algorithmicx, algpseudocode} 
\usepackage{graphicx}
\usepackage{cleveref}
\usepackage{color}
\usepackage{bbm}
\usepackage{centernot}
\usepackage{thmtools}
\usepackage{thm-restate}
\usepackage{enumitem} 
\usepackage{wrapfig}
\usepackage{enumitem}
\usepackage{comment}
\usepackage[numbers, sort]{natbib}
\usepackage{multibib}
\usepackage{url}

\setlist{itemindent=1.5ex}
  
\renewcommand{\epsilon}{\varepsilon}

\newcommand{\E}{{\mathbf E}}

\renewcommand{\sf}[1]{\mathsf{#1}}

\renewcommand{\L}{\Lambda}

\DeclareMathOperator{\poly}{poly}

\newcommand{\wRIP}{\sf{\gamma}(n)\text{-}\sf{RIP}}
\newcommand{\pRIP}{\sf{RIP}}

\newcommand{\secref}[1]         {Section~\ref{sec:#1}}
\newcommand{\seclabel}[1]    {\label{sec:#1}}

\newcommand{\figlabel}[1]   {\label{fig:#1}}
\newcommand{\figref}[1]         {Figure~\ref{fig:#1}}
\newcommand{\thmlabel}[1]   {\label{thm:#1}}
\newcommand{\thmref}[1]         {Theorem~\ref{thm:#1}}
\newcommand{\lemlabel}[1]   {\label{lem:#1}}
\newcommand{\lemref}[1]         {Lemma~\ref{lem:#1}}

\renewcommand{\eqref}[1]          {Eq.~\ref{eq:#1}}

\newcommand{\defn}{\emph}

\title{Efficient Rational Proofs with Strong Utility-Gap Guarantees\thanks
{A preliminary version of this paper will appear in the proceedings of the 11th International Symposium on Algorithmic Game Theory (SAGT) 2018. This work has been partially supported by NSF CAREER Award CCF 1553385,
CNS 1408695, CCF 1439084, IIS 1247726, IIS 1251137, CCF 1217708, by Sandia National Laboratories,
and by 
the European Union's 7th Framework Programme (FP7/2007-2013)~/~ERC grant agreement no. 614331.
BARC, Basic Algorithms Research Copenhagen, is supported by the VILLUM Foundation grant 16582.
}}

\urldef{\mailsa}\path|{jingchen, smccauley, shiksingh}@cs.stonybrook.edu|

\author{
Jing Chen\inst{1}
\and
Samuel McCauley\inst{2}
\and
Shikha Singh\inst{2}}

\institute{
Stony Brook University.
\email{jingchen@cs.stonybrook.edu} \and
Wellesley College. \email{\{samuel.mccauley, shikha.singh\}@wellesley.edu}.}

\begin{document}
\maketitle
\begin{abstract}
As modern computing moves towards smaller devices and powerful cloud platforms, 
more and more computation is being delegated to powerful service providers.
Interactive proofs are a widely-used model to 
design efficient protocols for verifiable computation delegation.

  Rational proofs are payment-based interactive proofs.
The payments are designed to incentivize the provers to give correct answers. 
If the provers misreport the answer then they incur a payment loss of at least $1/u$, where $u$ is the {\em utility gap} of the protocol. 

  In this work, we tightly characterize the power of rational proofs that are super efficient, that is, require only logarithmic
time and communication for verification. 
We also characterize the power of single-round rational protocols that require only logarithmic space and randomness
for verification. Our protocols have strong (that is, polynomial, logarithmic, and even constant) utility gap. 
Finally, we show when and how rational protocols can be converted to give the completeness and soundness guarantees of
classical interactive proofs.

\end{abstract}

\section{Introduction}\seclabel{intro}
Most computation today is not done locally by a
client, but instead is outsourced to third-party service providers in exchange
for money. Trading computation for money brings up two problems---(a)
how the client can guarantee correctness of the outsourced computation (without redoing the computation), and (b) how to design the payment scheme. The two
problems are closely related: ideally, we want the payment scheme to
be such that it incentivizes service providers to perform the computation
correctly. 

Interactive proofs (IP) are the most well-studied and widely-used 
theoretical framework to verify
correctness of outsourced computation\cite{bitansky2012succinct,
canetti2013refereed,
braun2013verifying,
goldwasser2008delegating,
rothblum2013interactive,
gur2015non,
kilian1992note,
kalai2015arguments,
cormode2012practical,
canetti2012two, 
chakrabarti2015verifiable,
cormode2011verifying,
daruki2015streaming}.  In an IP, a weak client
(or \defn{verifier}) interacts with powerful service providers (or \defn{provers}) to
determine the correctness of their claim. At the end, the verifier
probabilistically accepts or rejects the claim.\footnote{In classical interactive proofs there is no payment---simply acceptance or rejection.} 
Interactive proofs guarantee that, roughly speaking,
the verifier accepts a truthful claim with probability at least $2/3$ (\emph{completeness})
and no strategy of the provers can make the verifier accept a false claim with probability more than~$1/3$~(\emph{soundness}).\footnote{
More formally, given an input $x$ and a language $L$, if $x\in L$, 
the verifier accepts with probability at least $2/3$ ({\em completeness}); if $x\notin L$, 
then no
strategy of the provers can make the
verifier accept with probability more than $1/3$ (\emph{soundness}).}.  

Rational proofs are payment-based interactive proofs for computation outsourcing
which leverage the incentives of the service providers. 
In rational proofs, the provers
act {\em rationally} in the game-theoretic sense, that is,
they want to maximize their payment.
The payment is designed such 
that when the provers maximize their payment, they also end up giving the correct answer.
The model of rational proofs ($\sf{RIP}$) was introduced
by Azar and Micali in~\cite{azar2012rational}. Since then, 
many simple and efficient rational protocols have been designed%
~\cite{azar2013super,
guo2014rational,
guo2016rational,
zhang2014efficient,
hubavcek2014rationality,
campanelli2015sequentially,
ChenMcSi16}.

While rational proofs provide strong theoretical guarantees, there are two main barriers that separate them from what is often desired in practice.  First, many rational protocols require a polynomial-time verifier---but a ``weak'' client is unlikely to be able to spend (say) quadratic time or linear extra space on verification.   Second, many of these protocols strongly rely on the rationality of the provers.  An honest prover may receive only a fraction of a cent more than a dishonest prover,
yet a rational prover is assumed to be incentivized by that small amount. However,  
service providers may not always be perfectly rational. 

The goal of this paper is to give protocols that overcome these barriers. 

\paragraph{\bf Utility Gap.}
The strength of the guarantee provided by rational proofs is captured by the notion of {\em utility gap}.
 The high level idea
behind utility gap is that provers who are not perfectly rational may not
care about small losses in payments and may lazily give the incorrect answer.
If a rational protocol has a utility gap
of $u$, then the provers who mislead the verifier to an incorrect
answer are guaranteed to lose at least $1/u$. (This is under a normalized budget of $1$; if the budget is scaled up to $B$,
such provers can be made to lose at least $B/u$.)
Thus, protocols with small utility gap are sound even against 
provers with \defn{bounded rationality}; that is, provers who are only sensitive to large losses.

In this paper, we design efficient rational protocols with strong utility gap---that is,
polynomial, logarithmic, and even constant utility gap. 
In~\secref{ipvsrip}, we show when and how a noticeable utility gap of a rational protocol
can be utilized to achieve the strong completeness and soundness guarantees of a classical proof.

\paragraph{\bf Efficient Protocols.}
In this paper, we focus on designing rational protocols
with very small overheads in terms of verification time, space, communication cost
and number of rounds. In particular,
we design constant-round rational protocols
where the verification time and communication cost are logarithmic in the input size $n$. 
We also design single-round rational protocols that
have only logarithmic overhead on the verifier's use of space and randomness.
 
\subsection{Results and Contributions}
In this section, we summarize our results and contributions. 

\paragraph{\bf Super time-efficient rational proofs.}
We study the effect of different communication costs and an additional prover
on the power of rational proofs with a highly time-efficient verifier. 
The utility gap of these protocols is polynomial.
\begin{itemize}[leftmargin=*]
\item {\bf Constant communication.} 
We show that multiple provers do not add any power when the communication complexity of the protocol
is restricted to be extremely small---a constant number of bits.
That is, we show that the class of languages that
admit a multi-prover rational proof with a $O(\log n)$-time verifier and $O(1)$
communication is exactly $\sf{Uniform TC_0}$, which is the same as the 
power of single-prover version under the same costs~\cite{azar2013super, guo2014rational}.
$\sf{Uniform TC_0}$ is the class of constant depth, polynomial size uniform threshold circuits, 
that includes problems such as integer division, and iterated multiplication~\cite{allender1999permanent,
hesse2002uniform, hajnal1993threshold,allender1990power}.

\item {\bf Logarithmic communication.} We 
show that any rational proof with polynomial communication
can be simulated by a rational proof with logarithmic communication that uses an additional prover. Using this property,
we improve the communication complexity of Azar and Micali's~\cite{azar2013super} single-prover rational protocol 
and show that the class of languages that admit a two-prover rational proof with logarithmic communication is exactly the class
of languages decidable by a polynomial
time machine that can make polynomially many queries in parallel to an $\sf{NP}$ oracle, denoted
 $\sf{P_{||}^{NP}}$.\footnote{For parallel oracle queries,
both notations $\sf{P_{||}^{NP}}$~\cite{wagner1990bounded} and $\sf{P^{||NP}}$~\cite{azar2013super}
are used in literature.}
This is an important class~(e.g., \cite{wagner1990bounded, jenner1993computing, buhrman1994functions, krentel1988complexity, mathon1979note}) that includes optimization problems such as maximum
clique, longest paths, and variants of the traveling salesman problem. 
\end{itemize}

\paragraph{\bf Super space-efficient rational proofs.}
We achieve even better utility gap guarantees for the setting where the verifier's
use of space and randomness is super-efficient. In particular, we exactly characterize the class of single-round rational proofs with $\gamma(n)$
utility gap and logarithmic space and randomness as the class
of languages decidable by a polynomial-time machine that makes $O(\gamma(n))$ queries to an $\sf{NP}$ oracle,
denoted 
$\sf{P_{||}^{NP[\gamma(n)]}}.$ Even when $\gamma(n) = O(1)$ this bounded-query class is still sufficiently powerful and contains
many of the optimization problems mentioned above.

Thus, highly space-efficient rational protocols with strong guarantees against imperfectly rational provers 
can solve many important optimization problems.
\paragraph{\bf Rational proofs with completeness and soundness guarantees.}
Finally, we closely compare the two proof systems---rational and classical.  We 
construct a condition on the expected payments of rational proofs which, if satisfied,
turns them into a classical interactive proof with completeness and soundness guarantees.
We first show how to convert a payment-based protocol
for a language $L$ to an accept-reject protocol (without payments) for $L$ such that
the expected payment of the former 
is exactly the probability with which the verifier accepts in the latter. We use this 
to prove that if the expected payments of all inputs $x\in L$ are noticeably far away from that of 
 all inputs $x \notin L$, the rational protocol can be converted to a classical interactive protocol.

\subsection{Additional Related Work}
Azar and Micali~\cite{azar2013super} also characterize the classes $\sf{Uniform TC_0}$ and $\sf{P_{||}^{NP}}$.
Their characterization of $\sf{P_{||}^{NP}}$
requires polynomial communication, which we improve to logarithmic
using a second prover.
We also note that all protocols in~\cite{azar2013super}
have a polynomial utility gap (under a constant budget).

Rational arguments, super-efficient rational proofs where the prover is restricted to be polynomial time,
were introduced by Guo et al.~\cite{guo2014rational}.
Rational arguments for all languages in $\sf{P}$ were given in~\cite{guo2016rational}.
Campanelli and Rosario~\cite{campanelli2015sequentially} study sequentially composable rational proofs.
Zhang and Blanton~\cite{zhang2014efficient} design protocols to outsource matrix multiplications to a rational cloud.

The model of multi-prover rational interactive proofs was introduced by Chen et al.~\cite{ChenMcSi16},
where they study the power of the model in its full generality (that is, polynomial-time verifier
and polynomial communication). In this paper, we restrict our focus on
the power of multi-prover rational proofs when the verifier's running time and communication are restricted to be logarithmic.

Different variants of the rational-proof models have also been studied.
Chen et al.~\cite{chen2017rational} consider rational proofs
where the rational provers are {\em non-cooperative}~\cite{chen2017rational}. Inasawa and Kenji~\cite{inasawa2017rational}
consider rational proofs where the verifier is 
also rational and wants to minimize the payment to the provers. 

Interestingly, the logarithmic-space verifier studied in this paper also happens to be a {\em streaming algorithm}, that is, 
the verifier does not need to look again at any input or message bits out of order. Thus, our space-efficient
rational proofs are closely related to the work on streaming interactive proofs~\cite{cormode2012practical,chakrabarti2015verifiable,
cormode2011verifying,daruki2015streaming}. 

Refereed games is another multi-prover interactive-proof model that leads to game-theoretic characterizations of various complexity classes~\cite{chandra1976alternation, feige1990noisy,
feige1997making, feige1992multi, reif1984complexity,
feigenbaum1995game,koller1992complexity}. The model of refereed games requires at least
one honest prover. 

\section{Preliminaries}\label{sec:prelim}

We begin by reviewing the model of rational proofs~\cite{ChenMcSi16, azar2012rational}.

Let $L$ be a language, $x$ be an input string and $n= |x|$.
An {\em interactive protocol} is a pair $(V, \vec{P})$, where
$V$ is the {\em verifier} and
$\vec{P} = (P_1,\ldots,P_{p(n)})$ is the vector of {\em provers}, and $p(n)$ a polynomial in $n$.
The goal of the verifier is to determine if $x\in L$.
In general, the verifier runs in time polynomial in $n$ and uses polynomial space
as well. In section~\secref{logtime}, the verifier's
running time is $O(\log n)$. In \secref{logspace},
the verifiers may use polynomial time but are restricted to use $O(\log n)$
space and randomness. 
The provers are computationally unbounded.\footnote{While the model allows for extremely powerful provers, those considered in this paper essentially only need to be powerful enough to determine if $x\in L$ or $x \notin L$.}

The verifier can communicate with each prover privately, but no two provers can communicate with each other.
In a \defn{round}, either each prover sends a message to the verifier, or the
verifier sends a message to each prover, and these two cases alternate.  Without loss of generality, provers send the first round of messages.
The first bit of the first round is the \defn{answer bit}, denoted by $c$, and indicates whether $x\in L$; that is, $x \in L$ iff $c=1$. We define the \defn{communication} of the protocol to be the maximum number of total bits transmitted (summed over all provers and all rounds) during the protocol.
The length of each message and the number of rounds in a protocol
are bounded above by the communication cost. 

Let $r$ be the random string used by $V$.
Let $\vec{m}$ be the vector of all messages exchanged.
At the end, the verifier
computes the total payment to the provers, given by
a payment function $R(x, r, \vec{m})$. 
We restrict the verifier's budget to be constant, that is, $R \in [0, 1]$ for convenience. 
We may use negative payments to emphasize penalties but they can shifted to be non-negative.
The protocol (including the payment function~$R$) is public knowledge.

The verifier outputs the answer bit $c$ at the end of the protocol---thus the verifier always agrees with the provers.

Each prover $P_i$ can choose a {\em strategy} $s_{ij}:\{0,1\}^*\rightarrow\{0,1\}^*$
for each round $j$, which maps the transcript
he has seen up until the beginning of round $j$
to the message he sends in round $j$.
Note that $P_i$ does not send any
message when $j$ is even;  in this case $s_{ij}$ can be treated as a constant function.
Let $s_i = (s_{i1},\ldots, s_{ik})$ be the strategy vector of $P_i$
and $s = (s_1,\dots, s_{p(n)})$ be the strategy profile of the provers.
Given any input $x$, randomness $r$ and strategy profile $s$, we may write the vector $\vec{m}$ of messages exchanged in the protocol
more explicitly as $(V, \vec{P})(x, r, s)$.

The provers are {\em cooperative} and jointly act to maximize the total expected payment.
Thus, before the protocol starts, the provers pre-agree on a strategy profile $s$ that maximizes
 $   u_{(V, \vec{P})}(s; x) = {\mathbf E}_r [R(x, r, (V, \vec{P})(x, r, s))].$
When $(V, \vec{P})$ and $x$ are clear from the context, we write $u(s)$ for $u_{(V, \vec{P})}(s; x)$.
\begin{definition}[\cite{ChenMcSi16}]\label{def:mrip}
For any language $L$, an interactive protocol $(V, \vec{P})$ is a rational interactive proof protocol for $L$ if,
for any $x\in \{0, 1\}^*$ and any strategy profile $s$ of the prover(s)
such that $u(s) = \max_{s'} u(s')$, $c =1$ if and only if $x\in L$.
\end{definition}

Similar to classical proofs, single-prover rational interactive protocols, that is, when $\vec P = P$, are denoted by $\sf{RIP}$.
Multi-prover interactive protocols, where $\vec P = (P_1, \ldots, P_{p(n)})$ are denoted by $\sf{MRIP}$. In this
paper we study both single-prover and multi-prover rational proof protocols. 

We use $\poly(n)$ as a shorthand for a polynomial $n^k$, for some constant $k$.

\subsection{Utility Gap and Budget in Rational Proofs}

In the above definitions, we assume that a prover is fully rational, and will give the correct answer for \emph{any} increase in expected payment, no matter how small.  However, a prover may be lazy, and unwilling to give the correct answer unless the correct answer increases its payment by some minimum amount.  This consideration is particularly important in cloud-based systems, where the payment must be enough to offset the service providers' computational costs.  This consideration is analyzed in rational proofs using the budget and utility gap. 

A similar concept has been studied in classical interactive proofs, which  
provide completeness and soundness guarantees. In particular, they guarantee
that if $x\in L$, there exists $\vec P$ such that the verifier accepts $(V,
\vec P)$ with probability at least $c$, and and if $x\notin L$, for all $ \vec
P'$, the verifier accepts $(V, \vec P')$ with probability at most $s$ where
$c$ and $s$ are completeness and soundness parameters respectively, usually
constant with $0 < s < c< 1$.

The guarantee analogous to completeness and soundness in rational proofs
is that of utility gap. 
The notion of utility gap captures the payment loss incurred by provers
who misreport the answer bit. Since rational proofs are closed under
complement and provers can report $x\in L$ by sending $c=1$
or report $x\notin L$ by reporting $c=0$, unlike classical proofs,
completeness and soundness are not separately defined. Instead, 
utility gap captures the difference between completeness and soundness
for each instance $x$. 

In \secref{ipvsrip}, we show that any rational
proof protocol can be converted to one where the payments are $0$ or $1$,
where $1$ represents ``acceptance of claim $c$" and $0$ represents ``rejection of
claim $c$''. In such a protocol, the probability of acceptance is
then exactly equal to the expected payment of the provers. Informally, utility gap then is,
for a given $x$, the difference between
the probability that $V$ accepts $(V, \vec P)$
where $P$ makes an honest claim $c$ about $x$ and the probability that $V$ accepts a $(V, \vec P')$ where $\vec P'$ makes
a dishonest claim $c'$ about $x$.

\begin{definition}[\cite{ChenMcSi16}]\label{def:rewardgap}
Let $L$ be a language with a rational proof protocol $(V, \vec{P})$ 
and let $\gamma(n) \ge 0$. We say that $(V, \vec{P})$ has an {\em $\gamma(n)$-utility gap} if
for any input $x$ with $|x|=n$, any strategy profile
$s$ of $\vec{P}$ that maximizes the expected payment, and
 any other strategy profile $s'$, where the answer bit $c'$ under $s'$ does not match
the answer bit $c$ under $s$, i.e., $c'\neq c$, then
$u(s) - u(s') > {1}/{\gamma(n)}$.
\end{definition}

\paragraph{\bf Relationship between utility gap and budget.} 
The \defn{budget} is the total expected payment that a verifier can give in a protocol. 

Utility gap and budget are closely related.
To study utility gaps consistently, we maintain a fixed $O(1)$ budget.\footnote{
  In contrast, Azar and Micali~\cite{azar2013super} maintain a polynomial-size budget.}
This is because 
utility gap scales naturally with the payment---a polynomial utility gap
under a constant budget is the same as a constant utility gap under a sufficiently-large polynomial budget.

\subsection{Analyzing Computational Costs of Rational Proofs}
Our primary focus in this paper is analyzing the various computational costs of rational interactive proofs.
The different parameters fall into two categories.

\paragraph{\bf Verification costs.} A verifier has three main resources: running time, space usage and its randomness.

In \secref{logtime}, we focus on time-efficient $O(\log n)$ time verifiers. 
Thus, their space
and randomness is also $O(\log n)$. We denote the class of languages that have time-efficient 
RIP protocols, that is, protocols with a $O(\log n)$ time
verifier as $\sf{RIP}^t$. 
Multi-prover notation $\sf{MRIP}^t$ is analogous. 
Similar to the literature on ``probabilistically checkable proofs of
proximity'' (PCPPs)~\cite{rothblum2013interactive,gur2015non,ben2005short,ben2006robust}, we assume that the verifier has random access to the input string and the proof tape. Thus, if the messages
sent by the provers is $C(n)$ bits, the verifier needs at least $O(\log C(n))$ time to index a random location of the 
transcript.

As a logarithmic verifier cannot even read the entire input,
it is difficult to obtain protocols with good utility gap guarantees using these verifier.
To achieve better utility gap, in \secref{logspace}, we restrict the verifier's space usage and randomness, instead of its running time
and consider verifiers that use $O(\log n)$ space and $O(\log n)$ randomness. We denote the class of languages that have an RIP protocol 
with space- and randomness-efficient verifiers, that is, verifiers with $O(\log n)$ space and $O(\log n)$ randomness as
$\sf{RIP}^{s,r}$.


\paragraph{\bf Protocol costs.} A rational interactive proof protocol has three main ingredients: communication cost, number of rounds
of interaction and utility gap.\footnote{The number of provers is an additional parameter in MRIP protocols, but we ignore this
so as not to overload notation. All the MRIP protocols in this paper have two provers and all the upper bounds
work even with polynomially many provers.}

In~\secref{logtime}, we study the effect of varying the communication complexity of a protocol on its power when we have a logarithmic
time verifier. The number of rounds in all the protocols in the paper is $O(1)$. 

We denote the class of languages that have an RIP protocol with communication cost $C(n)$, number of rounds $k(n)$ and utility gap $\gamma(n)$ as $\sf{RIP}[C(n), k(n), \gamma(n)]$.
The multi-prover version is defined similarly.

\section{Verification in Logarithmic Time}\seclabel{logtime}

In this section we consider \emph{time-efficient} verifiers that run in time logarithmic in the input size.  We show that for time-efficient verifiers, access to multiple provers is fundamentally linked to the communication cost of the protocol: any single-prover protocol with high communication costs can be reduced to a communication-efficient multi-prover protocol.  On the other hand, multiple provers give no extra power for communication-efficient protocols.

Since the utility gap of all the protocols in this section is polynomial in $n$, we drop it from the notation
for simplicity. Thus, an RIP protocol with a $O(\log n)$-time verifier that has communication complexity $C(n)$ and round complexity $k(n)$ is
denoted as $\sf{RIP}^t[C(n), k(n)]$.

  \paragraph{\bf Constant communication.}\label{sec:constant-com}
We first show that multiple provers do not increase the power of a rational proof system when the communication complexity of the protocol is
very small, that is, only $O(1)$ bits. Recall
that with a single prover, 
$\sf{RIP}^t[O(\log n), O(\log n)] = \sf{RIP}^t[O(\log n), O(1)]= \sf{Uniform TC}_0 $~\cite{azar2013super, guo2014rational}.

  \begin{restatable}{theorem}{charpoly}\label{thm:char-poly}
$\sf{MRIP}^t[O(1), O(1)] 
= \sf{Uniform TC}_0.$
\end{restatable}

\begin{proof}
The lower bound follows directly from the single prover result~\cite{azar2013super,guo2014rational}.

Now, we prove that $\sf{MRIP}^t[O(1), O(1)] \subseteq \sf{Uniform TC}_0$. 

Let $L$ be
a language with a $k$-round MRIP protocol $(V, \vec{P})$ where $V$ runs in $O(\log n)$ time, and the transcript of $(V, \vec{P})$ has size $O(1)$,
where $k$ is a constant.

Note that the strategy profile $s$ of the provers $\vec{P}$ for a protocol 
with $O(1)$ communication can be specified in $O(1)$ bits. Thus, there
can be at most $O(1)$ possible strategy profiles for the provers to choose from.
We first construct a circuit that
decides $L$ and then show that the circuit can be
simulated by a $\sf{Uniform TC}_0$ machine. 

  We construct the gates in independent blocks (i.e. there are no wires between two gates in different blocks). We denote the $i$th block by $G_i$ for $1\le i\le t$ for some constant $t$. The purpose of $G_i$ is to ``try out'' strategy profile $s^i$.
In particular, the output of the block $G_i$ is the expected payment over all possible coin flips of the verifier when the strategy followed
by the provers is $s^i$. 
  All blocks finally output their solution (the expected payment) to 
  a single $\max$ gate that finds the maximum over the expected
payments. 

The structure inside a block $G_i$ is as follows: for each possible randomness $r$ of the verifier 
we have an input wire to the block gate (note that there are at most polynomially many $r$). Given
$r$, executing the strategy of the provers in a step by step manner (using the truth table) 
can be simulated by a depth $k$ circuit using using $\mbox{AND}$, $\mbox{OR}$, and $\mbox{NOT}$ gates. Thus, for each $r$, a
constant sized circuit can compute the final payment.

Finally, a $\mbox{SUM}$ gate at the end of $G_i$ can sum over
  the payments to compute the expected payment.\footnote{We could normalize by dividing by the number of possible $r$, but this scaling is unnecessary as we only care about relative payments.} This final expectation is the output of the block $G_i$. 

The final $\mbox{MAX}$ gate over the output all $G_i$ gives the
maximum possible expected reward. If the first bit of the corresponding strategy matches the first bit of the MRIP protocol's transcript,
then the circuit outputs $1$, else $0$.

We note that the above circuit structure can be simulated by a constant-depth uniform threshold circuit
since $\mbox{SUM}$ gates and $\mbox{MAX}$ gates with
polynomial input wires can be implemented using $\sf{Uniform TC}_0$ circuits~\cite{azar2013super}.  
\end{proof}

  \paragraph{\bf Logarithmic and polynomial communication.}\label{sec:poly-com}
We characterize the power of MRIP protocols with $O(\log n)$-time verification,
when the communication complexity of the protocol is logarithmic and polynomial in $n$.

  \begin{restatable}{theorem}{charlogpoly}\thmlabel{thm:char-log-poly}
$\sf{MRIP}^t[\poly(n), \poly(n)] 
=  \sf{MRIP}^t[O(\log(n)), O(1)] =\sf{P_{||}}^{\sf{NP}}.$
\end{restatable}

Azar and Micali~\cite{azar2013super} characterized the class $\sf{P_{||}^{NP}}$
in terms of single-prover rational proofs with $O(\log n)$ verification and $O(\poly(n))$ communication.
In particular, they proved that $\pRIP [O(\poly (n)), O(1)] = \sf{P_{||}^{NP}}$.

To prove~\thmref{thm:char-log-poly}, we first show that using
two provers reduces the communication complexity of the RIP protocol for $\sf{P_{||}^{NP}}$ exponentially.
In fact, we show prove a more general statement---any MRIP protocol (thus any RIP protocol as well) with
a logarithmic time verifier and polynomial communication can be simulated using two provers,
five rounds and logarithmic communication.

\begin{restatable}{lemma}{reducecommunication}\lemlabel{reducecommunication}
A MRIP protocol with $p(n)$ procers, $k(n)$ rounds, verification complexity $T(n)$, and communication complexity of $C(n)$
can be simulated by an MRIP protocol with $2$ provers, 5 rounds, verification complexity $O(T(n) + \log C(n))$ and communication
  complexity $O(T(n) + \log C(n))$.
\end{restatable} 

\begin{proof}
Let $(V, \vec{P})$ be the MRIP protocol for a language $L$ with $p(n)$ provers where $V$'s running time is $T(n)$  
and $C(n)$ bits of communication are exchanged over $k(n)$ rounds. Without loss of
generality, suppose each message is of length $\ell(n)$.
Note that $k(n)\le C(n)$ and $\ell(n) \le C(n)$. 
We shift and scale the payment function $R$ of $(V,\vec{P})$ such that $R \in [0,1]$.
The 2-prover 5-round MRIP protocol $(V', P_1', P_2')$
in~\figref{rip-to-mrip} simulates $(V, \vec{P})$.

\begin{figure}[tbhp] 
\centering 
\fbox{
\begin{minipage}{0.96\textwidth} 
{\normalsize
\vspace{0.0ex}
\noindent{}
 
For any input string $x$ of length $n$, the protocol $(V', P_1', P_2')$ works as follows.
\begin{enumerate}[nolistsep, leftmargin=10pt]
\item $P_1'$ sends $m_1$ to $V'$, where $m_1$ is the message sent by $P_1$ in the first round of $(V, \vec{P})$ according to the best strategy profile $s$ of $\vec{P}$. 
$V'$ outputs $c$, the first bit of $m_1$ at the end of the protocol.
\item $V'$ generates the random string $r$ used by $V$ in $(V, \vec{P})$ 
and sends it to $P_1'$.
\item\label{round:effectt} $P_1'$ sends a string $\tilde{m}$, which is a concatenation of bits accessed by $V$ in order in the transcript $\vec{m} = (V, \vec{P})(x, r,s)$.
\item\label{round:check} $V'$ chooses a round $j$ from $\{1,2 \ldots, k(n)\}$, a prover index $i \in \{1, \ldots, p(n)\}$ and a bit index $k$ from $\{1, 2, \ldots, \ell(n)\}$
uniformly at random.
\item $V'$ simulates $V$ using $\tilde{m}$ and sends all messages sent by $V$ to $P_i$ up to round $j-1$ to $P_2$. 
\item\label{round:commits} $P_2$ sends a bit $b$ to $V'$, where $b$ represents the $k$th bit of the round-$j$ message sent by $P_i$ in $(V, \vec{P})$.
\item $V'$ simulates $V$ to check if $V$ ever accesses the $k$th bit of $P_i$'s round $j$ message in $\vec{m}$. 
If $V$ does not, then the protocol ends and $R'=0$. 
\item\label{round:matchingbit} Finally, $V'$ computes the payment $R'$ as follows.
\begin{enumerate}[noitemsep, nolistsep, leftmargin=10pt]
%
	\item If $b$ does not match the $k$th bit of $P_i$'s round-$j$ message in $\tilde{m}$, $R' = -1$.

	\item Else, $V'$ computes the payment $R$ in $(V, \vec{P})$, and $R' = {R}/({2 C(n)})$.
\end{enumerate}
\end{enumerate}
}
\end{minipage}
}
\caption{MRIP protocol simulating an RIP protocol for $L \in \sf{RIP}[T(n),C(n)]$.}
\figlabel{rip-to-mrip}
  \vspace{-.2in}
\end{figure}

The string $\tilde{m}$ in step~\ref{round:effectt} is the {\em effective transcript} of the protocol $(V, \vec{P})$. 
Since $V$ runs in time $T(n)$, for a given randomness $r$, $V$ can access at most $T(n)$ bits from the $C(n)$-size
transcript $\vec{m}$ of the protocol $(P,V)$. Thus, $|\tilde{m}| \le T(n)$. 

Furthermore, since any index $i,j$ where $1 \le i \le \ell(n)$
and $1 \le j \le k(n)$ is of size at most $O(\log C(n))$, the communication complexity of the protocol $(V', P_1', P_2')$
is $O( T(n)+ \log C(n))$.
  Similarly, $i,j,k$ can be randomly selected in $O(\log C(n))$ time, leading to total time $O(T(n) + \log C(n))$.

We now prove correctness of the protocol in~\figref{rip-to-mrip}. 
Note that $P_2'$ commits to a strategy profile $s'$ of the provers $\vec{P}$ in step~\ref{round:commits}. We consider two cases.\\
{\bf Case 1.} For some randomness $r$, suppose $P_1'$ and $P_2'$ do not agree on the effective transcript $\tilde{m}$. 
Then without loss of generality, 
there exist indices $i,j, k$ such that the corresponding bit is part of the effective transcript, and
the verification in step~\ref{round:matchingbit} fails with $R'=-1$. The probability 
that $V'$ chooses such indices $i,j, k$ in step~\ref{round:check} is at least  $1/C(n)$.
Thus, the expected payment of the provers is at most:

\[ -1 \left( \frac{1}{C(n)}\right) + \frac{R}{2 C(n)} \left( \frac{C(n) - 1}{C(n)} \right)
\le \frac{1}{C(n)} \left( \frac{R}{2} - 1 \right) < 0,\]
where the last inequality follows from the fact that $R \le 1$. If $P_1'$ and $P_2'$
are consistent on $\tilde{m}$ and $r$  (keeping the rest of their strategy the same)
they can improve their expected payment since their payment under $r$ would be least $0$. Thus, this case does not occur
under the best strategy profile of
$P_1'$ and $P_2'$.\\ 
{\bf Case 2.} $P_1'$ and $P_2'$
agree on the effective transcript $\tilde{m}$ for every randomness $r$, but the
answer bit $c'$ under the strategy $s'$ committed by $P_2$ for $\vec{P}$ is incorrect. 

For a given randomness $r$ and indices $i, j$, and $k$ such that $V$ accesses $k$th bit of the round-$j$ message of $P_i$
in $(V,\vec{P})$, $R'= {R(s',x)}/({2C(n)})$, where
$R(s',x)$ is the payment of protocol $(V,\vec{P})$ under strategy $s'$. By the correctness
and utility gap of $(V,\vec{P})$, we know that the expected payment $u(s',x) + 1/\poly(n)< u(s,x)$, where $s$ is the best strategy of $\vec{P}$.
Thus, in this case, $P_1'$ and $P_2'$ lose a polynomial amount if they use strategy $s'$ instead of $s$.
\end{proof}

\lemref{reducecommunication} demonstrates the importance of two provers over
one to save on communication cost in rational proofs. 

%
\begin{corollary}
$\sf{RIP}^t[O(\poly(n)), O(1) ] =  
\sf{P_{||}^{NP}}  \subseteq \sf{MRIP}^t[ O(\poly(n)), O(\poly(n)] \subseteq \sf{MRIP}^t[O(\log n), O(1) ]. $
\end{corollary}

To complete the proof~\thmref{thm:char-log-poly}, we prove the following upper bound.

\begin{restatable}{lemma}{logpolyupper}\lemlabel{log-poly-upper}
$\sf{MRIP}^t[ O(\log(n)), O(1)] \subseteq \sf{P_{||}^{NP}}.$
\end{restatable}
The proof of \lemref{log-poly-upper} is similar to the proof of $\sf{MRIP}[\poly(n), \poly(n)] \subseteq \sf{EXP^{||NP}}$ in~\cite{ChenMcSi16}. We include it here for completeness.

\begin{proof}
Fix a language $L\in \sf{MRIP}(\log(n), \log(n), O(1))$ and let $(V, \vec{P})$ be an MRIP protocol for $L$.
Since $V$ runs in $O( \log n)$ time, there exists a constant $k$ such that, for any two payments $R$ and $R'$ 
computed by $V$ for some input of length $n$ and some randomness are such that $R\neq R' \Rightarrow |R-R'|\geq \frac{1}{{n^k}}.$

Moreover, since $V$ uses $O(\log n)$ coin flips, there exists another constant $k'$ such that, when a payment appears with positive probability under some input of length $n$, it must appear with probability at least $\frac{1}{{n^{k'}}}$.

Therefore, for any input $x$ of length $n$ and any two strategy profiles $\tilde{s}$ and $\tilde{s}'$ of the provers,
if the expected payments $u(\tilde{s}; x)$ and $u(\tilde{s}'; x)$ are different, then
\begin{equation}\label{equ:u}
|u(\tilde{s}; x) - u(\tilde{s}'; x)|\geq \frac{1}{{n^{k+k'}}}.
\end{equation}

Consider the following deterministic oracle Turing machine $M$:
Given any input $x$ of length $n$, it divides the interval $[0,1]$ into $2 {n^{k+k'}}$ subintervals of length $\frac{1}{2 {n^{k+k'}}}$.
For any $i\in \{1,\ldots, 2 {n^{k+k'}}\}$, the $i$th interval is $[(i-1)/2 {n^{k+k'}}, i/2 {n^{k+k'}}]$.
$M$ then makes $4 {n^{k+k'}}$ queries of the form $(i, j)$, where $i\in \{1,..., 2 {n^{k+k'}}\}$ and $j\in \{0,1\}$.
%
For each query $(i, j)$, if $j=0$ then the corresponding question is ``whether there exists a strategy profile $\tilde{s}$ of the provers such that
$u(\tilde{s}; x)$ is in the $i$th interval''; and if $j=1$ then the corresponding question is ``whether there exists a strategy profile $\tilde{s}$ such that
$u(\tilde{s}; x)$ is in the $i$th interval {\em and the first bit sent by $P_1$ is $c=1$}''.
Note that all queries are non-adaptive.
We say that interval $i$ is {\em non-empty} if the query $(i, 0)$ is answered $1$, and {\em empty} otherwise.

Given the answers to all the queries, $M$ finds the highest index $i^*$ such that the interval $i^*$ is non-empty.
It accepts if $(i^*,1)$ is answered $1$, and rejects otherwise. 
Given correct oracle answers, we show that $M$ decides $L$.

For by the definition of MRIP, there exists a strategy profile whose expected payment is 
non-negative and thus in $[0, 1]$.
Thus there exists an interval $i$ such that $(i, 0)$ is answered~$1$.
Also by definition, the best strategy profile $\tilde{s}$ has the highest expected payment, and thus $u(\tilde{s};x)$ falls into interval $i^*$.

By Inequality \ref{equ:u}, any strategy profile $\tilde{s}'$ with $u(\tilde{s}'; x)< u(\tilde{s}; x)$ has $u(\tilde{s}'; x)$ not in interval $i^*$, since the difference between
$u(\tilde{s}'; x)$ and $u(\tilde{s}; x)$ is larger than the length of the interval.
And so all strategy profiles $\tilde{s}'$ with $u(\tilde{s}'; x)$ in interval $i^*$ satisfies
$u(\tilde{s}'; x) = u(\tilde{s}; x)$, that is, they are all the best strategy profiles of the provers.
$P_1$ must send the same first bit $c$ under all such strategy profiles,
$c=1$ if and only if $x\in L$, and there does not exist any other strategy profile whose expected payment 
falls into interval $i^*$ but the first bit sent by $P_1$ is different from $c$.
Thus the answer to $(i^*, 1)$ always equals $c$, and $M$ accepts iff $c=1$.

We now show that the oracle queries can be answered by an $\sf{NP}$ oracle.
Since the communication complexity is at most $O(\log n)$ 
a strategy profile has size polynomial in $n$.
%
Thus, an NP machine can guess a strategy profile $\tilde{s}$, simulate the protocol, 
and compute the expected payment $u(\tilde{s}; x)$. 
%
%
\end{proof}

\section{Verification in Logarithmic Space}\seclabel{logspace}

The protocols in~\secref{logtime} have a polynomial utility gap. For a constant budget this means that the provers who
mislead the verifier to an incorrect answer lose at least $1/\poly(n)$ of their expected payment. 

As utility gap is analogous to the soundness gap in classical proofs, which is constant (independent of $n$), it is desirable to have rational protocols with constant utility gap as well.

Constant utility gap is difficult to achieve when the verifier is $O(\log n)$ time and cannot even read the entire input. This is true
even for classical proofs with a $O(\log n)$-time verifier where the soundness conditioned is weakened
to design ``proofs of proximity''~\cite{rothblum2013interactive,gur2015non,ben2005short,ben2006robust}. In particular, the soundness guarantees of such proofs depend on how
far (usually in terms of hamming distance) the input string $x$ is from the language $L$.
We note that
all existing $O(\log n)$-time rational proofs~\cite{azar2013super,guo2014rational,guo2016rational} have polynomial utility gap (under a constant budget).

To design protocols with a strong utility gap such as logarithmic or constant, in this section we consider verifier's
that use only $O(\log n)$ space
and randomness.

Let $\gamma(n)$ be any polynomial-time computable function (given $1^n$) that is polynomially bounded. 
For example, $\gamma(n)$ can be a constant, $\log n$, or $\sqrt{n}$.
We prove the characterization in general for a utility gap of $\gamma(n)$.

\begin{theorem}
  \thmlabel{logspacechar}
Let $\sf{P_{||}^{\sf{NP[\gamma(n)]}}}$ be a polynomial-time Turing machine that can make $O(\gamma(n))$ non-adaptive queries to an $\sf{NP}$ oracle. This
class is equivalent to the class of languages that have a one-round RIP protocol with a logspace verifier, polynomial communication and $\gamma(n)$-utility gap.
That is,
\[ \sf{RIP}^{r,s} [\poly(n), 1, \gamma(n)] =  \sf{P}_{||}^{\sf{NP[\gamma(n)]}}. \]
\end{theorem}

First, we give a space-efficient RIP for the class $\sf{NP}$
using the log-space
interactive proof for the language given by Condon and Ladner~\cite{condon1995interactive} as a blackbox.

\begin{restatable}{lemma}{constantgapNP}\lemlabel{constantgapNP}
$\sf{NP} \in \sf{RIP}^{r,s} [\poly(n), 1, \gamma(n)] .$
\end{restatable}

\begin{proof}
Let $(V,P)$ denote the $1$-round log-space interactive proof for a language $L \in \sf{NP}$ given in~\cite{condon1995interactive}.
The one-round log-space RIP for $\L$, $(V', P')$ is given. $P'$ sends message $m'$
which is the concatenation of answer bit $c$ with a bit string $m$. If $c=0$, then $m$ can be a null string.
If $c=1$, then $m$ must be the message sent by $P$ in $(V,P)$. If $c=0$,
then $R=1/2$ and the protocol ends. Otherwise, $V'$ simulates $V$ using $m$ and if $V$ accepts, 
then $R=1$,
else $R=0$.

The verifier $V'$ uses the same space as $V$, that is,  $O(\log n)$. The communication of $(V', P')$
is the same as $(V,P)$, that is, polynomial in $n$.

We now argue correctness and show that the protocol has constant utility gap. Suppose $x \in L$ and $P$ sends a message
  with answer bit $c'=0$, then his expected payment is $1/2$. On the other hand, if $P$ sends $c=1$, his expected payment can be $1$ by
the completeness guarantee of $(V,P)$. Thus in this case the expected payment loss of $P$ is constant.
Now suppose $x \notin L$ and $P$ sends the answer bit $c'=1$. From the soundness guarantee of $(V,P)$, $V$ accepts with probability at most $1/3$, and thus the expected payment of $P$ is at most $1/3$. On the other hand,
  if $P$ sent $c=0$, his expected payment would have been $1/2$. Thus in this case $P$ loses a constant amount as well. 
\end{proof}

For the lower bound, we use a different but equivalent complexity class. 
Let $\mathbf{L}_{||}^{\sf{NP[\gamma(n)]}}$ be a logarithmic space machine that can make $O(\gamma(n))$ non-adaptive queries to an $\sf{NP}$ oracle.
Wagner~\cite{wagner1990bounded} showed that $ \mathbf{L}_{||}^{\sf{NP[\gamma(n)]}} = \sf{P_{||}^{NP[\gamma(n)]}}$.
 
%
 
  \begin{restatable}{lemma}{logspacelower}
$\sf{P_{||}^{NP[\gamma(n)]}}=\mathbf{L}_{||}^{\sf{NP[\gamma(n)]}} \subseteq \sf{RIP}^{r,s} [\poly(n), 1, \gamma(n)]$
\end{restatable}

\begin{proof}
Consider any language $L \in \mathbf{L}_{||}^{\sf{NP[\gamma(n)]}}$.
  Let $M$ be the logarithmic-space machine with at most $\gamma(n)$ nonadaptive accesses to an oracle $O$ for an $\sf{NP}$ language
that decides $L$. 
Without loss of generality, suppose
$M$ makes exactly $\gamma(n)\geq 1$  non-adaptive queries to $O$.
The RIP protocol for $L$ uses the RIP protocol for $\sf{NP}$, given in~\lemref{constantgapNP}, to simulate the oracle queries.


  \begin{figure}[H]

\centering
\fbox{
\begin{minipage}{0.96\textwidth}
{\normalsize
\vspace{0.5ex}
\noindent{}For any input $x$ of length $n$, the protocol $(V, P)$ works as follows. Let $R_n = 0$.
\begin{enumerate}[noitemsep, leftmargin=10pt]
\item $P$ sends a message $c, (c_1, m_1), (c_2, m_2), \ldots, (c_{\gamma(n)}, m_{\gamma})$ to $V$,
where $c$ is the answer bit of the entire protocol, and $c_i$ is the answer bit for the $i$th oracle query $q_i$ generated by $M$ and $m_i$
is the corresponding proof for $q_i$ based on~\lemref{constantgapNP}.  $V$ outputs $c$ at the end of the protocol.

\item $V$ simulates $M$ on $x$ until $M$ outputs queries $q_1,\dots, q_{\gamma(n)}$.
\item  For each
$i\in \{1,\ldots, \gamma(n)\}$, $V$ simulates $V'$ in the RIP protocol for $\sf{NP}$ in~\lemref{constantgapNP}. 
In particular, $V$
uses the message $m_i$ as the prover's message in the protocol of~\lemref{constantgapNP}. Let $c^*_i$ and $R^*_i$ be the answer bit and the payment in that protocol respectively.
$V$ returns $c^*_i$ as the oracle's answer for $q_i$,
and
updates the sum $R_n \leftarrow R_n + R^*_i$.
\item\label{step:final}
$V$ continues simulating $M$ till the end.
If $c$ does not match $M$'s output, then $R =-1$;
otherwise $R = R_n/\gamma(n)$.
\end{enumerate}
}
\end{minipage}
}
\caption{An RIP protocol for $\mathbf{L}_{||}^{\sf{NP[\gamma(n)]}}$.}
\figlabel{log-constgap}
\end{figure}

We now prove correctness of the protocol in~\figref{log-constgap}.
Note that an honest strategy of $P$, that is, reporting the correct answer bits $c, c_1, \ldots, c_{\gamma(n)}$,
and sending correct proof strings $m_i$ whenever $c_i =1$, leads to a payment $R \ge 1/2$ (this
is because of the payment-structure of the protocol for NP in~\lemref{constantgapNP}). 

Suppose $P$ reports the incorrect answer bit $c'$, then either (a) the output of $M$ in Step~\ref{step:final} does not match $c'$ and $R=-1$; or
(b) there exists an $\sf{NP}$ query $q_i$ such that the answer bit $c^*_i$ is incorrect.

In case (a), the expected payment loss of $P$ is at least $1/2+1 = 3/2>1/\gamma(n)$,
as $\gamma(n) \geq 1$. In case (b),
because the protocol in~\lemref{constantgapNP} has $O(1)$ utility gap,
the provers' expected payment loss in the overall protocol is 
$1/O(\gamma(n))$.
\end{proof}

To complete the proof of \thmref{logspacechar} we prove the following upper bound.
  \begin{restatable}{lemma}{logspaceupper}
 $\sf{RIP}^{r,s} [\poly(n), 1, \gamma(n)]\subseteq \sf{P_{||}^{NP[\gamma(n)]}}$
\end{restatable}

\begin{proof}
Given any $L \in \wRIP$, let $(V, P)$ be
the a 1-round RIP protocol for $L$ with $\gamma(n)$ utility gap, where $V$ uses $O(\log n)$
space and the communication complexity is $O(\poly(n))$. 

  Consider a polynomial-time Turing machine $M$ which can make $\gamma(n)$ nonadaptive accesses to an $\sf{NP}$ oracle $O$. 
Given any input $x$ of length $n$, $M$ divides $[0,1]$ into $2 \gamma(n)$ intervals, each of length $1/(2\gamma(n))$.
That is, the $i$th interval is $[i/2\gamma(n),(i+1)/2\gamma(n))$ for each $i \in \{1, \ldots, 2\gamma(n)-1\}$.

For each such interval, $M$ queries its oracle $O$:
\smallskip
\begin{enumerate}[noitemsep,nolistsep,leftmargin=*]
\item\label{step:exists} Does there exist a message $m$ sent by $P$ such that the expected payment $u_{(V, P)}({s}; x)$ is in 
the $i$th interval?
\item\label{step:ans}  Does there exist a message $m$ sent by $P$ such that the expected payment $u_{(V, P)}({s}; x)$ is in 
the $i$th interval and the corresponding answer bit $c=1$?
\end{enumerate}
\smallskip

$M$  makes $O(\gamma(n))$ non-adaptive queries
and clearly runs in polynomial time. 

We now show that it is sufficient for the oracle $O$ to be an $\sf{NP}$ machine.
The key point to note here is that protocol is one round and thus the size
of a provers' strategy is polynomial in $n$. Thus, an $\sf{NP}$ oracle can 
guess a strategy and compute a payment $R(s, x, (V, P))$. Since the verifier
uses $O(\log n)$ randomness, an $\sf{NP}$ machine can enumerate over
all possible polynomial coin flips and compute the expected payment $u_{(V,P)}(s,x)$.
If $u_{(V,P)}(s,x)$ is in the $i$th interval
the oracle returns $1$ to query~(\ref{step:exists}) and $0$ otherwise. Similarly, if
$u_{(V,P)}(s,x)$ is in the $i$th interval and $c=1$ for this strategy, the oracle responds $1$
to query~(\ref{step:ans}) and $0$ otherwise. 
Thus, $M$'s queries can be answered by an NP oracle.

Finally, 
$M$ finds the highest index $i^*$ such that the oracle returns $1$ to query~(\ref{step:exists}) with respect to the $i^*$. 
$M$ accepts if the oracle returns $1$ to query~\ref{step:ans} for the $i^*$th interval, and rejects otherwise.

To see why $M$ decides $L$ given correct answers to its oracle queries,
note that the maximum expected payment $u_{(V,P)}(s^*;x)$ that $P$ can get falls in the $i^*$th interval.
As $(V, \vec{P})$ has $\gamma(n)$-utility gap, 
by construction and definition of utility gap, all strategies with expected payments in
the $i^*$th interval must have the same answer bit $c$ as that in $s^*$.
Thus, $x\in L$ if and only if $c=1$, which occurs
if and only if the oracle's answer to query~\ref{step:ans} for interval $i^*$ is 1. 
\end{proof}

\section{Relationship Between Classical and Rational Proofs}\seclabel{ipvsrip}
In this section, we show under what conditions does a rational interactive proof
reduces to a classical interactive proof. The results in this section are stated in terms
of the multi-prover model (that is, $\sf{MRIP}$ and $\sf{MIP}$) which is more general,
and thus they also hold for the single prover model (that is, $\sf{RIP}$ and $\sf{IP}$).

To compare the two proof models, we explore their differences. In rational interactive proofs, the provers
are allowed to give an answer bit $c=1$ claiming $x\in L$ or $c=0$ claiming $x\notin L$.\footnote{Thus it is not surprising
that rational proofs are closed under complement.} 
In other words, the question is ``is $x \in L$?" and the rational provers can say ``yes'' or ``no'' based on their incentives.
 Furthermore, for a particular input $x$ of size $n$, if the provers' claim $c$ about $x$ is incorrect, they lose at least
a $1/\gamma(n)$, where $\gamma(n)$ is the utility gap.

On the other hand,
in classical proofs, the provers are only allowed to prove membership, that is, they are only allowed to prove $x\in L$.
Furthermore, given completeness and soundness parameters $c$ and $s$ respectively, where $0 \le s < c \le 1$,
 we have that ``for any $x \in L$'', there exists a strategy such that $V$ accepts with probability at least
$c$ and ``for any $x \notin L$'', for any strategy $V$ rejects with probability at most $s$. 
Thus, given $L$,
the guarantees are independent of $x$.

In this section, we show that a rational proof reduces to a classical proof. 
Intuitively, this happens when the utility gap guarantee of a rational protocol is made to hold for all $x$ and in particular, 
it is enforced to be the gap between the expected payments for all $x\in L$ and all $x\notin L$.

We first show that without loss of generality we can restrict the
payments of the provers in a rational proof protocol to be either $1$ or $0$, where $1$ corresponds to ``accept'' and $0$ to ``reject'' respectively. 

\begin{restatable}{lemma}{zeroone}\lemlabel{zeroone}
Any MRIP protocol $(V, \vec{P})$ with payment $R \in [0,1]$ 
  and utility gap $\gamma(n)$ 
  can be simulated
by a MRIP protocol $(V', \vec{P})$ with payment $R' \in \{0, 1\}$ 
  and utility gap $\gamma(n)/2$.
  In particular, for any strategy $s$ and any input $x$, 
  \[
    u_{(V,\vec{P})}(x;s)\leq u_{(V',\vec{P})}(x;s) \leq u_{(V,\vec{P})}(x;s) + \gamma(n)/2.
  \]
 $V'$ uses $1 + \lceil \log_2 \gamma(n)\rceil$ more random bits than $V$. 
\end{restatable}

\begin{proof}
  We will go one by one through the payments made in the rational protocol $(V, \vec P)$ and replace them with payments in $\{0,1\}$.  
  At a high level, $V'$ makes a payment
  $1$ in $(V', \vec{P'})$ with probability $R(x,r,\vec{m})$ where $R(x,r,\vec{m})$ is the payment made by $V$ in $(V, \vec{P})$ , and $0$ otherwise.

  Assume without loss of generality that each random string $r$ is a bit string that corresponds exactly to the result of $|r|$ coin flips.  Let $G = 2^{1 + \lceil \log_2 \gamma(n)\rceil}$; thus $1/G \leq 1/2\gamma(n)$. 

  We create a new protocol $(V', \vec{P}')$.  First, $V'$ runs the original protocol $(V, \vec{P})$ to obtain a transcript $\vec{m}$ (given the provers' strategy $s$) and randomness $r$
of $V$. Let $R(x,r,\vec{m})$ be the payment made by $V$.  Then, $V'$ flips $1 + \lceil \log_2 \gamma(n)\rceil$ extra coins; call this string $r'$.  Momentarily treat $r'$ as an integer. If $r' \leq \lceil G\cdot R(x,r,\vec{m})\rceil$ then $V'$ pays 1; otherwise $V'$ pays 0.  Let this final payment 
be denoted by $R(x,r\circ r',\vec{m})$.  

  Note that in the above, because $r'$ is uniformly selected from $G$ distinct values, for any $r$, the $V'$ pays 1 with probability
  \[
    \E_{r'} [R(x,r\circ r',\vec{m})] = \frac{\lceil G\cdot R(x,r,\vec{m})\rceil}{G};
  \]
  thus
  \[
    R(x,r,\vec{m}) \leq  \E_{r'} [R(x,r\circ r',\vec{m})]  \leq R(x,r,\vec{m}) + 1/G\leq R(x,r,\vec{m}) + 1/2\gamma(n).
\]

  The expected payment of the original protocol on input $x$ and transcript $\vec{m}$ is $\sum_r R(x,r,\vec{m}) \Pr(r)$.  
  The expected payment of the new protocol is (using independence of $r'$ and $r$)
  \begin{align*}
    \E_{r,r'} [R(x,r\circ r',\vec{m})] = &\sum_{r} \sum_{r'} R(x,r\circ r',\vec{m}) \Pr(r') \Pr(r)\\ 
&=  \sum_r \E_{r'}[R(x,r\circ r',\vec{m})] \Pr(r).
  \end{align*}
Substituting with the above, the expected payment is bounded by
  $$
    \sum_r R(x,r,\vec{m})\Pr(r) \leq \E_{r,r'} [R(x,r\circ r',\vec{m})] \leq 1/2\gamma(n) + \sum_r R(x,r,\vec{m})\Pr(r).
\qedhere$$
\end{proof}

Given any rational protocol with zero-one payments, we note that it immediately gives us an accept-reject protocol such that
for a given $x$, the probability that the verifier accepts is exactly the expected payment of the original protocol.
More formally let $(V, \vec P)$ be a rational protocol with $R \in \{0,1\}$
and utility gap $\gamma(n)$. Let $(V', \vec{P}')$ be defined as follows: $V'$ simulates $V$, ignores the answer bit $c$,
and if the payment in $(V, \vec{P})$ is $R=1$ then accept, else reject.

Thus, for a given input string $x$, the expected payment in $(V, \vec{P})$ is equal to the probability that $V'$ accepts in $(V', \vec{P}')$.
That is,
\begin{align}
u_{(V,\vec P)}(x; s) &= \E_r [R(x, r, (V, \vec{P})(x, r, s))] \notag =\sum_{r} \Pr(r \text{~$|$~} R(x, r, (V, \vec{P})(x, r, s))=1)\notag\\
&= \sum_{r} \Pr(r \text{~$|$~} V' \text{accepts $(V', \vec{P'})$})= \Pr(\text{$V'$ accepts $(V', \vec{P}')$}) \label{exp-prob}.
\end{align}
Furthermore, $(V', \vec{P}')$ satisfies the following instance-specific properties
similar to completeness and soundness in interactive proofs.
For any $x\in L$, let $s^*$ denote the optimal strategy of the provers $\vec P$, that is,
$s^*$ maximizes their expected payment. Then for $\vec{P'}$
following $s^*$, $V'$ accepts with probability exactly $c(x, n) = u_{(V,\vec P)}(x; s^*)$. Furthermore,
we know from the utility gap condition that for any $x \notin L$, for
any strategy $s'$, the probability
that $V'$ accepts is at most $ u_{(V,\vec P)}(x; s') < u_{(V,\vec P)}(x; s^*) - 1/\gamma(n)$,
that is, the probability that $V'$ accepts is at most $s(x, n) < c(x,n)- 1/\gamma(n)$.
Similar guarantees hold for any $x\notin L$.

However, if we want $(V', \vec{P}')$ to be an interactive proof protocol in the classical sense,
that is, with completeness and soundness guarantees that hold for all $x \in L$ and for all $x\notin L$ respectively,
we need to impose restrictions on the expected payment function of the rational protocol.

\begin{theorem}\thmlabel{main:toip} 
Let $(V, \vec P)$ be an MRIP protocol for a language $L$ such that 
\begin{equation}\label{condition}
\min_{x \in L} u_{(V,\vec P)}(x; s^*) > \max_{x \notin L}  u_{(V,\vec P)}(x; s^*) + \frac{1}{\gamma(n)}
\end{equation}
where $x$ is any input of length $n$, $s^*$ is
the strategy of the provers that maximizes their expected payment in $(V, \vec P)$ and $\gamma(n)$ is
any function such that $\gamma(n)>1$ and $\gamma= O(\poly(n))$. 
Then, $(V, \vec P)$ can be simulated by a MIP protocol for $L$. 
\end{theorem}

We prove this theorem in two parts. First, we show prove the following lemma
which proves~\thmref{main:toip} with weak completeness
and soundness guarantees. 

\begin{restatable}{lemma}{toip}\lemlabel{toip} 
Let $(V, \vec P)$ be an MRIP protocol for a language $L$ that satisfies the condition~\ref{condition} in~\thmref{main:toip}.
Then, $(V, \vec P)$ can be simulated by MIP protocol 
with completeness and soundness parameters $c(n)$ and $s(n)$ respectively such that $c(n) > s(n) + 1/2\gamma(n)$
and $c(n), s(n)\ge 0$.
\end{restatable}
\begin{proof}
Using \lemref{zeroone}, without loss of generality, let the payment of $(V, \vec P)$
  be $R \in \{0, 1\}$. 
  Since the expected reward under each strategy is changed by at most $1/2\gamma(n)$, the condition of \thmref{main:toip} is still satisfied with $\gamma(n) \gets 2\gamma(n)$.
In the MIP protocol $(V', \vec P)$, $V'$ simulates $V$, ignores the answer bit $c$,
and if the payment in $(V, \vec{P})$ is $1$, then accepts, else rejects.
For a given $x$, the expected payment in $(V, \vec{P})$ is equal to the probability that $V'$ accepts; see~Equation~\ref{exp-prob}.

Define $c(n) = \min_{x \in L} u_{(V,\vec P)}(x, r, s^*)$
and $s(n) = \max_{x \notin L} u_{(V, \vec P)} (x, r, s^*)$. Then, by definition, we have $c(n), s(n) \geq 0$
and $c(n) > s(n) + 1/2\gamma(n)$.

We now show that $c(n)$ and $s(n)$ are the completeness and soundness parameter of the MIP $(V', \vec{P'})$ respectively.
For any $x\in L$,
there exists $\vec{P'}$ where $\vec{P'}$ uses a strategy $s^*$, such that
the probability $V'$ accepts is exactly $u_{(V,\vec P)}(x, s^*) \geq c(n)$.
For any $x\notin L$, for all $\vec{P'}$ using any strategy $s$, 
the probability $V'$ accepts is exactly $u_{(V,\vec P)}(x, s) \leq u_{(V,\vec P)}(x, s^*) \leq s(n)$.
Thus, $(V', \vec{P}')$ is an MIP with completeness and soundness parameters $c(n)$ and $s(n)$ respectively. 
\end{proof}

We amplify the ``gap'' of an MIP by repeating 
the protocol sufficiently many times and then using Chernoff bounds. The techniques are mostly standard, although the parameters must be set carefully to deal with the case $s(n) = 0$.

\begin{restatable}{lemma}{amplify}\lemlabel{amplify}
Given an MIP protocol for a language $L$, with completeness $c(n) > 0$ and soundness $s(n)\geq 0 $ such
  that $c(n) > s(n) + 1/\gamma'(n)$ for some $\gamma'(n)>1$ and $\gamma' =O (\poly(n))$, can be converted to an MIP protocol for $L$
  with completeness at least $1- 1/\poly(n)$ and soundness at most $1/\poly(n)$.
\end{restatable}

\begin{proof}

  We repeat the MIP protocol 
  $
    \rho(n) = 96(\log n)\gamma'(n)^2/c(n)
  $
  times and accept if more than $\tau(n) = \rho(n)c(n)(1 - 1/4\gamma'(n))$
  of the instances end in accept.   

Let the random indicator variable $X_i$ be $1$ if the verifier in the $i$th repetition accepts, otherwise $X_i=0$.
  Let $X= \sum_{i=1}^{\rho(n)} X_i$ be the total number of accepts. 

%
Consider an $x\in L$. Then if provers use their best strategy of the original MIP protocol in each iteration, we have
\[
  \E[X] = \E\left[\sum_{i=1}^{\rho(n)} X_i\right] = \sum_{i=1}^{\rho(n)} \E\left[X_i\right] \ge \sum_{i=1}^{\rho(n)} c(n) = c(n)\rho(n).
\]
  Using Chernoff bounds, we obtain\footnote{This uses a slight extension of Chernoff bounds that uses a bound on the expectation rather than the expectation itself; see Exercise 4.7 in~\cite{MitzenmacherUpfal05} for example.} 
\begin{align*}
  \Pr(X <\tau(n)) 
  &= \Pr\left(X <  \left(1 - \frac{1}{4\gamma'(n)}\right) c(n)\rho(n)\right) \\
  &\leq e^{- \frac{c(n)\rho(n)}{32\gamma'(n)^2}} < 1/n.
\end{align*}

Now consider an $x\notin L$. For any strategy of the provers, by the soundness guarantee 
  at most $s(n)$ of the original protocol (and using linearity of expectation as above) we have
  $
  \E(X) \leq s(n)\rho(n) < \rho(n)(c(n) - 1/2\gamma'(n))$. 
  Note $c(n) - 1/2\gamma'(n) > c(n)/2$ and $\tau(n) > \rho(n)(c(n)-1/2\gamma'(n))(1 + 1/4\gamma'(n))$. 
  Then we can use the following  Chernoff bound 
\begin{align*}
  \Pr(X >\tau(n)) &\leq \Pr\left(X > \left(1 + \frac{1}{4\gamma'(n)} \right)\rho(n)(c(n) - 1/2\gamma'(n))\right)\\
  &\leq e^{-\rho(n)(c(n) - 1/2\gamma'(n))/48\gamma'(n)^2} \leq e^{-\rho(n)c(n)/96\gamma'(n)^2} \leq 1/n.
\end{align*} 

  The same analysis extends to any $1/\poly(n)$ instead of $1/n$ when $\rho(n)$ is increased by a constant. 
\end{proof}

\begin{remark}
The repetition of the MIP protocol to amplify its completeness and soundness guarantee
  used in~\lemref{amplify} is not efficient as it blows up the number of rounds. There exist 
more efficient techniques to amplify IP guarantees by parallel repetition that can be used instead; for example, see~\cite{goldreich1998modern, bellare1993randomness, raz1998parallel,feige1994two}.
\end{remark}

\bibliographystyle{plain}

\end{document}